\documentclass[runningheads, envcountsame]{llncs}
\usepackage[T1]{fontenc}
\usepackage{comment}

\usepackage{amssymb}
\usepackage{amsmath}
\usepackage{graphicx}
\usepackage{xcolor}
\usepackage{enumerate}
\usepackage[linesnumbered,  ruled,  vlined	]{algorithm2e}
\SetKwInOut{Input}{Input}
\SetKwInOut{Output}{Output}
\usepackage{subcaption}
\usepackage{tikz}
\usetikzlibrary{calc}
\usetikzlibrary{arrows,automata}
\usetikzlibrary{arrows.meta}
\usetikzlibrary{decorations.pathreplacing, calligraphy}
\usetikzlibrary{decorations.pathmorphing}
\usetikzlibrary{shapes.geometric}
 \tikzstyle{br} = [decorate, ultra thick, decoration = {calligraphic brace}]

\newcommand{\dcut}{{\sc $d$-Cut}}

\newcommand{\mmc}{{\sc Maximum Matching Cut}}
\newcommand{\pmc}{{\sc Perfect Matching Cut}}

\newcommand{\NP}{{\sf NP}}

\newcommand{\ssi}{\subseteq_i}
\newcommand{\si}{\supseteq_i}

\newtheorem{observation}[theorem]{Observation}

\counterwithin{nclaim}{theorem}

\definecolor{nicered}{RGB}{204,0,0}
\definecolor{lightblue}{RGB}{153,204,255}
\tikzstyle{vertex}=[thin,circle,inner sep=0.cm, minimum size=1.7mm, fill=black, draw=black]
 \tikzstyle{svertex}=[thin,circle,inner sep=0.cm, minimum size=1.3mm, fill=black, draw=black]
 \tikzstyle{bvertex}=[thin,circle,inner sep=0.cm, minimum size=1.7mm, fill=lightblue, draw=lightblue]
 \tikzstyle{rvertex}=[thin,circle,inner sep=0.cm, minimum size=1.7mm, fill=nicered,draw=nicered]
 \tikzstyle{evertex}=[thin,circle,inner sep=0.cm, minimum size=1.7mm, fill=none,draw=black]
 \tikzstyle{edge}=[thick, draw = gray]
 \tikzstyle{tedge}=[ultra thick, draw = black]
 \tikzstyle{tredge}=[ultra thick, draw = nicered]
 \tikzstyle{tbedge}=[ultra thick, draw=lightblue]
 \tikzstyle{redge}=[thick, draw = nicered]
 \tikzstyle{bedge}=[thick, draw = lightblue] 
 \tikzstyle{gedge}=[thick, draw = nicegreen] 
 \tikzstyle{brace} = [decorate, ultra thick, decoration = {calligraphic brace}]

\begin{document}

\title{Finding $d$-Cuts in Probe $H$-Free Graphs}

\titlerunning{Finding $d$-Cuts in Probe $H$-Free Graphs}

\author{Konrad K. Dabrowski\inst{1}\orcidID{0000-0001-9515-6945} 
\and Tala Eagling-Vose\inst{2}\orcidID{0009-0008-0346-7032}
\and  Matthew Johnson\inst{2}\orcidID{0000-0002-7295-2663} 
\and Giacomo~Paesani\inst{3}\orcidID{0000-0002-2383-1339}
\and Dani\"el Paulusma\inst{2}\orcidID{0000-0001-5945-9287}}

\authorrunning{K.K. Dabrowski, T. Eagling-Vose, M. Johnson, G. Paesani, D. Paulusma}

\institute{Newcastle University, Newcastle, UK \email{konrad.dabrowski@newcastle.ac.uk} \and
Durham University, Durham, UK \email{\{tala.j.eagling-vose,matthew.johnson2,daniel.paulusma\}@durham.ac.uk}\and
Sapienza University of Rome, Rome, Italy \email{ paesani@di.uniroma1.it}}

\maketitle             

\begin{abstract}
For an integer $d\geq 1$, the \dcut\ problem is that of deciding whether a graph has an edge cut in which each vertex is adjacent to at most~$d$ vertices on the opposite side of the cut. The $1$-{\sc Cut} problem is the well-known {\sc Matching Cut} problem. The \dcut\ problem has been extensively studied for $H$-free graphs. We extend these results to the probe graph model, where we do not know all the edges of the input graph. 
For a graph~$H$, a partitioned probe $H$-free graph $(G,P,N)$ consists of a graph $G=(V,E)$, together with a set $P\subseteq V$ of probes and an independent set $N=V\setminus P$ of non-probes
such that we can change $G$ into an $H$-free graph by adding zero or more edges between vertices in $N$.
For every graph~$H$ and every integer $d\geq 1$, we completely determine the complexity of \dcut\ on partitioned probe $H$-free graphs. 
\end{abstract}

\section{Introduction}\label{s-intro}

When studying computationally hard problems for special graph classes, it is natural to generalize polynomial-time results for certain graph classes to larger graph classes. In particular, consider a graph $H'$ and an induced subgraph $H$ of $H'$. The class of {\it $H$-free} graphs (class of graphs that do not contain $H$ as an induced subgraph) is contained in the class of $H'$-free graphs. Say an \NP-complete problem $\Pi$ is polynomial-time solvable on $H$-free graphs. Is $\Pi$ also polynomial-time solvable on $H'$-free graphs? This question leads to complexity studies for a wide range of graph problems where the goal is to obtain {\it complexity dichotomies} that tell us for exactly which graphs $H$ a certain \NP-complete problem is polynomial-time solvable, and for which graphs $H$ it stays \NP-complete.

We follow this line of research, but also assume that we do not know all the edges of the input graph. Before explaining the latter in more detail, we first introduce the problem that we study. 
Consider a connected graph $G=(V,E)$. A subset $M\subseteq E$ is an {\it edge cut} of $G$ if it is possible to partition $V$ into two non-empty sets $B$ ({\it blue} vertices) and $R$ ({\it red} vertices) in such a way that $M$ is the set of all edges with one end-vertex in $B$ and the other in~$R$.
Now, for an integer $d\geq 1$, if every blue vertex has at most $d$ red neighbours, and every red vertex has at most $d$ blue neighbours, then the edge cut $M$ is said to be a {\em $d$-cut} of $G$. See also Figure~\ref{fig-t-d-cut}.
The \dcut\ problem is that of deciding whether a connected graph has a $d$-cut. A $1$-cut is also called a {\it matching cut}, and the $1$-{\sc Cut} problem is better known as {\sc Matching Cut}.
For all $d\geq1$, $d$-{\sc Cut} is \NP-complete~\cite{Ch84,GS21}. 
Graphs with matching cuts were introduced in 1970 by Graham~\cite{Gr70} in the context of number theory;
for other applications see~\cite{ACGH12,FP82,GPS12,PP01}.

\begin{figure}[t]
\centering
\scalebox{.7}{
\begin{tikzpicture}
\coordinate (B1) at (-2,1);
\coordinate (B2) at (-1,1);
\coordinate (B3) at (0,1);
\coordinate (B4) at (1,1.5);
\coordinate (B5) at (2,1);
\coordinate (R1) at (-2,-1.5);
\coordinate (R2) at (-1,-1.5);
\coordinate (R3) at (0,-1);
\coordinate (R4) at (1,-1);
\coordinate (R5) at (2,-1.5);  
\draw[fill=gray!20!white](-0.5,-2) rectangle (2.5,2);
\draw (R4)--(B2)--(B3)--(B5)--(B4)(B1)--(B4)--(B3)--(R1)(B3)--(R5)(R2)--(B4)--(R3)(R4)--(B5)--(R5)(R5)--(R2)(R1)--(R3)--(R4)--(R2)--(R3)(B1)--(R3);
\draw[color=blue!50!white,thick](B1)--(B2)(R1)--(R2)(R1)--(B1)--(R2)--(B2)--(R1);
\draw[fill=blue] 
(B1) circle [radius=3pt]
(B2) circle [radius=3pt]
(B3) circle [radius=3pt]
(B4) circle [radius=3pt]
(B5) circle [radius=3pt];
\draw[fill=red] 
(R1) circle [radius=3pt]
(R2) circle [radius=3pt]
(R3) circle [radius=3pt]
(R4) circle [radius=3pt]
(R5) circle [radius=3pt];
\node[right] at (2.5,0) {$P$};
\end{tikzpicture}}
\caption{A graph $G$ with a set $P$ of probes. The set $F$ is the set of dashed edges. The blue-red colouring corresponds to a $2$-cut of $G$ and a $3$-cut in $G+F$.}\label{fig-t-d-cut}
\vspace*{-0.5cm}
\end{figure}
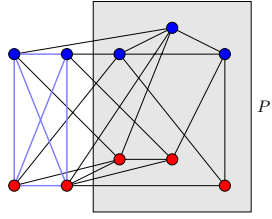

\medskip
\noindent
{\bf Our Focus.} 
We consider the classical probe graph model, which was introduced by Zhang et al.~\cite{ZSFCWKB94} in 1994 to deal with partial information in genome research. In this model,
 the complete set of neighbours is only known for {\it some} vertices of the input graph~$G$. These vertices form the set $P$ of {\it probes}. The other vertices of $G$ form the set $N$ of {\it non-probes}. As we do not know the adjacencies between vertices in $N$, the set $N$ is an independent set in $G$. However, in the probe graph model we also assume there exists a ``certifying'' set $F$ of edges between (some of) the non-probes such that $G+F$ has some known global structure; again see Figure~\ref{fig-t-d-cut}. In our paper, $G+F$ is $H$-free. Note that $G[P]$ is already $H$-free.

So, a {\it partitioned probe $H$-free} graph $(G,P,N)$ consists of a graph $G=(V,E)$, a set $P\subseteq V$ of probes and an independent set $N=V\setminus P$ of non-probes, such that $G+F$ is $H$-free for some edge subset $F\subseteq \binom{N}{2}$. Any $H$-free graph is also (partitioned) probe $H$-free: take $P=V$ and $N=\emptyset$. Hence, (partitioned) probe $H$-free graphs contain all $H$-free graphs, and any \NP-completeness results for $H$-free graphs carry over to partitioned probe $H$-free graphs. We therefore ask:

\medskip
\noindent
\textit{For which $H$, does $d$-{\sc Cut} stay polynomial-time solvable on probe $H$-free graphs?}

\medskip
\noindent
As such, our paper belongs to a recent systematic study of graph problems on probe $H$-free graphs. This study was initiated by Brettell et al.~\cite{BOPPRL25} for \textsc{Vertex Cover}, whereas the previous literature on probe graphs aimed to characterize and recognize classes of probe graphs; see e.g.~\cite{BGL07,CCKLP09,CKKLP05,GL04,GMM11}.
For example,  if $H=P_4$, then probe $H$-free graphs can be recognized in polynomial time~\cite{CKKLP05}.
However, for most other graphs~$H$, the complexity of recognizing probe $H$-free graphs is still unknown. Hence, for our algorithms, we assume that $P$ and $N$ are part of the input, that is, we will consider partitioned probe $H$-free graphs. 

We will also consider two related problems: {\sc Maximum Matching Cut} and {\sc Perfect Matching Cut} on probe $H$-free graphs. The first is to decide if a connected graph has a matching cut of at least $k$ edges for some integer~$k$. 
 The second is to decide if a connected graph has a {\it perfect matching cut}, that is, an edge cut that is a perfect matching. This problem is also \NP-complete~\cite{HT98}.

\medskip
\noindent
{\bf Known Results.} For two vertex-disjoint graphs $G_1$ and $G_2$, let $G_1+G_2=(V(G_1)\cup V(G_2),E(G_1)\cup E(G_2))$. We let $sG$ be the disjoint union of $s$ copies of $G$. We write $G_1\ssi G_2$ if $G_1$ is an induced subgraph of~$G_2$.
Let $C_s$ denote the cycle on $s$ vertices, $P_t$ the path on $t$ vertices, and $K_{1,r}$ the star on $r+1$ vertices. The graph $K_{1,3}$ is known as the {\it claw}.
Let $H^*_1$ be the ``H''-graph, which has vertices $u,v,w_1,w_2,x_1,x_2$ and edges $uv, uw_1,uw_2,vx_1,vx_2$.
For $i\geq 2$, let $H_i^*$ be obtained from $H_1^*$ by subdividing $uv$ exactly $i-1$ times.
See  Figure~\ref{fig-examples}. 

\begin{figure}[t]
\centering
\hspace*{1cm}
 \scalebox{.7}{
\begin{minipage}{0.4\textwidth}
\begin{tikzpicture}
\coordinate (S1) at (-2,0);
\coordinate (S2) at (-1,0);
\coordinate (S3) at (1,0);
\coordinate (S4) at (2,0);
\coordinate (P1) at (-1.5,-1);
\coordinate (P2) at (-0.5,-1);
\coordinate (P3) at (0.5,-1);
\coordinate (P4) at (1.5,-1);
\draw[fill=black]
(S1) circle [radius=3pt]
(S2) circle [radius=3pt]
(S3) circle [radius=3pt]
(S4) circle [radius=3pt]
(P1) circle [radius=3pt]
(P2) circle [radius=3pt]
(P3) circle [radius=3pt]
(P4) circle [radius=3pt]
(P1)--(P2)--(P3)--(P4);
\draw[dotted] (S2)--(S3);
\draw [decorate,decoration={brace,amplitude=5pt,mirror,raise=2ex}]   (S4)--(S1) node[rotate=0,midway,yshift=2.5em]{$i$ vertices};
\end{tikzpicture}
\end{minipage}%
\hspace*{1cm}
\begin{minipage}{0.25\textwidth}
\begin{tikzpicture}
\coordinate (S) at (0,2);
\coordinate (A1) at (-1,1);
\coordinate (B1) at (0,1);
\coordinate (C1) at (1,1);
\draw[fill=black]
(S) circle [radius=3pt]
(A1) circle [radius=3pt]
(B1) circle [radius=3pt]
(C1) circle [radius=3pt];
\draw (A1)--(S)--(B1)(S)--(C1);
\end{tikzpicture}
\end{minipage}%
\hspace*{1cm}
\begin{minipage}{0.4\textwidth}
\begin{tikzpicture}
\coordinate (S1) at (2,1);
\coordinate (S2) at (2,0);
\coordinate (S3) at (2,-1);
\coordinate (X1) at (1,0);
\coordinate (X2) at (0,0);
\coordinate (P1) at (-1,1);
\coordinate (P2) at (-1,0);
\coordinate (P3) at (-1,-1);
\draw[fill=black]
(S1) circle [radius=3pt]
(S2) circle [radius=3pt]
(S3) circle [radius=3pt]
(X1) circle [radius=3pt]
(X2) circle [radius=3pt]
(P1) circle [radius=3pt]
(P2) circle [radius=3pt]
(P3) circle [radius=3pt]
(S1)--(S3)(P1)--(P3)(S2)--(X1)(X2)--(P2);
\draw[dotted] (X1)--(X2);
\draw [decorate,decoration={brace,amplitude=5pt,mirror,raise=2ex}]   (S2)--(P2) node[midway,yshift=2.5em]{$i$ edges};
\end{tikzpicture}
\end{minipage}}
\caption{The graphs $sP_1+P_4$, $K_{1,3}$ and $H^*_i$, from left to right.}
\label{fig-examples}  
\vspace*{-0.4cm}
\end{figure}
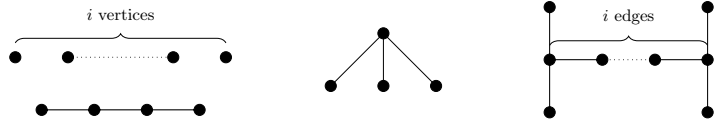

In Theorems~\ref{t-0}--\ref{t-2} we present the state-of-art for \dcut, \pmc\ and \mmc\ for $H$-free graphs. Only Theorem~\ref{t-2} is a full dichotomy.
The references in Theorem~\ref{t-0} are explained in~\cite{LMPS24} except for the recent result that $2$-{\sc Cut} is \NP-complete for claw-free graphs~\cite{AELPS}; note the jump in complexity from $d=1$ to $d=2$ for $H=3P_2$ and $H=K_{1,3}$. For $d\geq 2$, the only three non-equivalent open cases in Theorem~\ref{t-0} are $H=2P_4$, $H= P_6$ and $H=P_7$ (see also~\cite{LMPS24}).
The references in Theorem~\ref{t-1} are explained in~\cite{LPR24}.

\begin{theorem}[\cite{AELPS,Bo09,Ch84,FLPR23,LL23,LMO25,LMPS24,LPR22,LPR23a,Mo89}]\label{t-0} Let $H$ be a graph and $d\geq 1$.\\[-0.5cm]
\begin{itemize}
\item 
If $d=1$, then \dcut\ on $H$-free graphs
is polynomial-time solvable if $H\ssi sP_3+S_{1,1,2}$, $sP_3+P_4+P_6$, or $sP_3+P_7$ for some $s\geq 0$; and
\NP-complete if $H\si K_{1,4}$, $P_{14}$, $2P_7$, $3P_5$, $C_r$ for $r\geq 3$, or $H_i^*$ for  $i\geq 1$.\\[-8pt]
\item 
If $d\geq 2$, then \dcut\ on $H$-free graphs is polynomial-time solvable if $H\ssi  sP_1+P_3+P_4$ or $sP_1+P_5$ for some $s\geq 0$; and \NP-complete if $H\si K_{1,3}$, $3P_2$, $C_r$ for $r\geq 3$, or $H_i^*$ for  $i\geq 1$.
\end{itemize}
\end{theorem}

\begin{theorem}[\cite{FLPR23,LL23,LT22,LPR23a}]\label{t-1}
\pmc\ on $H$-free graphs is 
 polynomial-time solvable if $H\ssi sP_4+S_{1,2,2}$ or $sP_4+P_6$ for some $s\geq 0$; and
\NP-complete if $H\si K_{1,4}$, $P_{14}$,  $2P_7$, $3P_6$, $C_r$ for $r\geq 3$ or $H_j^*$ for $j\geq 1$.
\end{theorem}

\begin{theorem}[\cite{LPR24}]\label{t-2}
\mmc\ on $H$-free graphs is polynomial-time solvable if $H\ssi sP_2+P_6$ for some $s\geq 0$; and \NP-complete otherwise.
\end{theorem}

\noindent
{\bf Our Results.}
We combine \NP-completeness results from Theorems~\ref{t-0}--\ref{t-2} with new polynomial and hardness results (shown in Section~\ref{s-poly} and~\ref{s-np}, resp.) to prove: 

\begin{theorem}\label{t-dicho}
For a graph~$H$, the following four complete dichotomies hold:\\[-15pt]
\begin{itemize}
\item {\sc $1$-Cut}, \pmc\ and \mmc\ on partitioned probe $H$-free graphs are  polynomial-time solvable if $H\ssi sP_1+P_4$ for some $s\geq 0$; and \NP-complete otherwise;\\[-8pt]
\item for $d\geq 2$, \dcut\ on partitioned probe $H$-free graphs is polynomial-time solvable if $H\ssi P_1+P_4$; and \NP-complete otherwise.
\end{itemize}
\end{theorem}

\noindent
From Theorems~\ref{t-0} and~\ref{t-dicho}, it follows that \dcut\ becomes harder for {\em probe} $H$-free graphs (if ${\sf P} \neq \NP$) even if $H=2P_2$ for $d\geq 1$ and $H=4P_1$ for $d\geq 2$.

\section{Preliminaries and Basic Results}\label{s-prelim}


Let $G=(V,E)$ be a graph.
We let $N_G(v) = \{u \in V~|~uv \in E\}$ be the \emph{(open) neighbourhood} of~$v$ and $N_G[v] = N_G(v) \cup \{v\}$ be the \emph{closed neighbourhood} of~$v$. Let $S\subseteq V$. We write $G[S]$ to denote the subgraph of $G$ induced by $S$. A vertex $v \notin S$ is \emph{complete} to $S$ if $N(v) \supseteq S$, and $v$ is \emph{anti-complete} to $S$ if $N(v)\cap S=\emptyset$. Let $S'\subseteq V$ with $S'\cap S=\emptyset$.
If every vertex of $S$ is complete (anti-complete) to $S'$, then $S$ is {\it complete} ({\it anti-complete}) to $S'$. 

In our paper we also define some other probe graph classes. For example, we may say that a graph $G$ with a set $P$ of probes and a set $N$ of non-probes is {\em probe split} if there exists a set $F\subseteq \binom{N}{2}$ such that $G+F$ is a split graph (a graph whose vertex set can be partitioned into a clique and an independent set).

We now recall some colouring terminology for $d$-cuts from~\cite{LMPS24}  
that is commonly used in the context of matching cuts (see, e.g.~\cite{LPR22}). 
A {\it red-blue colouring} of a graph $G$ colours every vertex of $G$ either red or blue.
For $d\geq 1$, a red-blue colouring is a {\it red-blue $d$-colouring} if every blue vertex has at most $d$ red neighbours, every red vertex has at most $d$ blue neighbours, and $G$ has at least one blue vertex and at least one red vertex. See Figure~\ref{fig-t-d-cut} for red-blue $d$-colourings for $d=2$ and $d=3$. For some $d\geq 1$, a red-blue $d$-colouring is {\em perfect} if and only if every red vertex has exactly $d$ blue neighbours and vice versa.
This gives us the following straightforward observation (in the case of perfectness we focus on $d=1$: a perfect $1$-cut is a perfect matching cut).

\begin{observation}[\cite{LMPS24}]\label{o-cut-colouring}
For every $d\geq 1$, a connected graph $G$ has a (perfect) $d$-cut if and only if $G$ has a (perfect) red-blue $d$-colouring.
\end{observation}

Let $d\geq 1$. Let $G=(V,E)$ be a connected graph and $X, Y \subseteq V$ be disjoint sets. A \emph{red-blue $(X, Y)$-$d$-colouring} of $G$ is a red-blue $d$-colouring of $G$ that colours all the vertices of $X$ red and all the vertices of $Y$ blue. 
We say that $(X,Y)$ is a {\it $d$-precoloured pair} of $G$. We will usually ``guess'' such a pair $(X,Y)$ as the starting point in our algorithms.
On a $d$-precoloured pair $(X,Y)$ we can safely apply the following two rules exhaustively. 

\begin{enumerate}[\bfseries{R1.}]
\item Return {\tt no} (i.e. $G$ has no red-blue $(X,Y)$-$d$-colouring) if a vertex $v\in V$ is adjacent to $d+1$ vertices in $X$ as well as to $d+1$ vertices in $Y$.
\item Let $v\in V\setminus (X\cup Y)$. If $v$ is adjacent to $d+1$ vertices in $X$, then put $v$ in $X$, and if $v$ is adjacent to $d+1$ vertices in $Y$, then put $v$ in $Y$.
\end{enumerate}

\noindent
Afterwards, we either returned {\tt no}, or we obtained two new sets $X'\supseteq X$ and $Y'\supseteq Y$. In the latter case we say that we have {\it colour-processed} $(X,Y)$ {\it into} $(X',Y')$. By construction, every vertex of $V\setminus (X'\cup Y')$ is adjacent to at most $d$ vertices of $X'$ and to at most $d$ vertices of~$Y'$. 
The next lemma shows that we can work safely with $(X',Y')$ instead of $(X,Y)$.

\begin{lemma}[\cite{LMPS24}]\label{l-process}
Let $G$ be a connected graph with a precoloured pair $(X,Y)$. It is possible, in polynomial time, to either colour-process $(X,Y)$ into a pair $(X',Y')$ such that $G$ has a red-blue $(X,Y)$-$d$-colouring if and only if it has a red-blue $(X',Y')$-$d$-colouring, 
or to find that $G$ has no red-blue $(X,Y)$-$d$-colouring.
\end{lemma}

\section{Polynomial-Time Results}\label{s-poly}

In this section, we show our polynomial-time results for \mmc, \pmc\ and \dcut\ $(d\geq 1)$. That is, we show that \mmc\ (and thus $1$-{\sc Cut}) and \pmc\ are polynomial-time solvable on $(sP_1+P_4)$-free graphs, and {\sc $d$-Cut}, for $d\geq 2$, is polynomial-time solvable on $(P_1+P_4)$-free graphs. 

Our proofs are based on combining colour-processing with the observation that we can guess the closed neighbourhood of any set of size at most some constant~$c$: this will lead to only $\mathcal{O}(2^cn^{cd})$ branches, due to the fact that any vertex can have at most $d$ neighbours of the opposite colour. In our algorithms we choose a constant number of constant-sized sets in such a way that afterwards the whole input graph is coloured.

For our first result we need the following lemma.

\begin{lemma}\label{lem:mc-induction}
    For a graph $H$, the following two statements hold:
    \begin{enumerate}[(i)]
        \item If \mmc\ is polynomial-time solvable on partitioned probe $H$-free graphs, then
         it is polynomial-time solvable on partitioned probe $(P_1+H)$-free graphs.\\[-8pt]

        \item If \pmc\ is polynomial-time solvable on partitioned probe $H$-free graphs, then it is polynomial-time solvable on partitioned probe $(P_1+H)$-free graphs.    
    \end{enumerate}
\end{lemma}
    
\begin{proof}
We first prove (i). Suppose that \mmc\ is polynomial-time solvable on partitioned probe $H$-free graphs.
    Let $G=(V,E)$ be a partitioned probe $(P_1+H)$-free graph with a set $P \subseteq V$ of probes and a set $N = V\setminus P$ of non-probes. 
    If there is an edge subset $F'\subseteq \binom{N}{2}$ such that 
    $G+F'$ is $H$-free, then we can solve \mmc\ on $(G,P,N)$ in polynomial time by assumption. 
    Now suppose that for every edge subset $F'\subseteq \binom {N}{2}$, there is some $Q_{F'} \subseteq V(G)$ such that $G[Q_{F'}]$ is isomorphic to $H$ in $G+F'$. As $G$ is probe $(P_1+H)$-free, there exists an edge subset $F\subseteq \binom{N}{2}$ such that $G+F$ is $(P_1+H)$-free. We let $Q=Q_F$.

    For any red-blue $1$-colouring, of $G$, every vertex $v \in V(G)$ has at most one neighbour with a different colour than $v$, so there are $\mathcal{O}(n)$ red-blue colourings of $N_{G}[v]$. We branch on all $\mathcal{O}(n^{|H|})$ red-blue colourings of the closed neighbourhood of $Q$. As $H$ is a fixed graph, the number of branches is polynomial. Moreover,
the red-blue colouring in each branch corresponds to a pair $(X,Y)$ of $G$ in which every vertex of $X$ is red and every vertex of $Y$ is blue. Our goal is to find, 
for each such pair $(X,Y)$, a red-blue $(X,Y)$-$1$-colouring of $G$ that has as many edges with both a blue and a red end-vertex as possible; we say that such edges are {\it non-monochromatic}. We take the maximum number of non-monochromatic edges over all pairs $(X,Y)$ as our output. By Observation~\ref{o-cut-colouring}, we find a maximum matching cut in this way (should $G$ have a matching cut).

Let $(X,Y)$ be the pair corresponding to a red-blue colouring of the closed neighbourhood of $Q$.
By Lemma~\ref{l-process}, we can colour-process $(X,Y)$ into a pair $(X',Y')$ such that afterwards we either find that $G$ has no red-blue $(X,Y)$-$1$-colouring, or we may consider the pair $(X',Y')$ instead. In the former case, we discard the pair $(X,Y)$. Suppose that we are in the latter case. Note that from the nature of the rules {\bf R1} or {\bf R2}, any red-blue $(X,Y)$-$1$-colouring of $G$ with maximum number of non-monochromatic edges corresponds to a red-blue $(X',Y')$-$1$-colouring of $G$ with maximum number of non-monochromatic edges. 

A key observation is that any uncoloured vertex of $G$ belongs to $N$. This can be seen as follows.
If $u$ is an uncoloured vertex of $G$ that belongs to $P$, then $u$ is not in the closed neighbourhood of $Q$ (as otherwise we would have coloured $u$ already). As $u$ belongs to $P$, it follows that $u$ is not incident to an edge in $F$. We also recall that $G[Q]$ is isomorphic to $H$ in $G+F$. 
Hence, $Q\cup \{u\}$ induces a $P_1+H$ in $G+F$, contradicting the fact that $G+F$ is $(P_1+H)$-free.

As $N$ is an independent set of $G$, the above implies that the set of uncoloured vertices of $G$ is independent. This means that we may apply Lemma~21 from~\cite{LPR24} to find a red-blue $(X',Y')$-$1$-colouring with the maximum number of non-monochromatic edges in polynomial time.

\medskip
\noindent
We now prove (ii) by proceeding in exactly the same way as above. At some point, we must consider a pair ($X',Y')$, and we have a set of uncoloured vertices that is an independent set. 
Let $u$ be such an uncoloured vertex. Due to the colour-processing, $u$ has at most one blue neighbour and at most one red neighbour. If $u$ has only one neighbour, we must give $u$ the colour opposite to the colour of its neighbour. We check if the resulting red-blue colouring contains no vertex with two neighbours of the opposite colour. If that is the case, we discard the original pair $(X,Y)$ that corresponds to $(X',Y')$. 
Otherwise, we let $U$ be the remaining set of uncoloured vertices, which all have exactly one blue neighbour and exactly one red neighbour. It remains to check whether there exists a perfect matching~$M$ between the vertices of $U$ and its open neighbourhood in $G$. If so, we colour each vertex in $U$ blue if its matched neighbour in $M$ is red, and vice versa. We then check if the resulting red-blue colouring of $G$ is a perfect red-blue $1$-colouring. If not, we discard $(X,Y)$, and otherwise we are done. It takes polynomial time to find a perfect matching or to conclude that it does not exist. Hence, statement (ii) of the lemma has also been proven. \qed
\end{proof}

\noindent
We now prove a second lemma.

\begin{lemma}\label{lem:SPMC-p4}
    \mmc\ and \pmc\ are polynomial-time solvable on partitioned probe $P_4$-free graphs.
\end{lemma}

\begin{proof}
    Let $G$ be a probe $P_4$-free graph with a set $P\subseteq V$ of probes and a set $N=V\setminus P$ of non-probes. Hence, there exists an edge subset $F\subseteq \binom{N}{2}$ such that $G+F$ is $P_4$-free. 
    The latter implies that $G+F$ has a dominating 
    edge~$uv$, which is an edge such that every other vertex in $G+F$ is adjacent to at least one of $u,v$. The fact that such a dominating edge exists follows directly from using the alternative definition of $P_4$-free graphs as cographs.

    We now consider the edge $uv$ in $G+F$.
    As $F\subseteq \binom{N}{2}$, the only vertices that are anti-complete to $\{u,v\}$ in $G$ must belong to $N$, so these vertices form an independent set (as $N$ is independent). This means that we can use exactly the same arguments as in the proof of Lemma~\ref{lem:mc-induction}: we just replace $Q$ by $\{u,v\}$.
    
    For \pmc, we can proceed in a similar way as above. We can also use the fact that 
    probe $P_4$-free graphs have clique-width at most~$4$~\cite{CKKLP05}
    and {\sc Perfect Matching Cut} can be expressed in MSO$_1$, which means that we may apply the meta-theorem of Courcelle,  Makowsky and Rotics~\cite{CMR00}.\qed
\end{proof}

\noindent
We are now ready to prove our first main result of this section.

\begin{theorem}\label{t-mmc}
\mmc\ and \pmc\ are polynomial-time solvable on partitioned probe $(sP_1+P_4)$-free graphs for all $s\geq 0$.
\end{theorem}

\begin{proof}
By combining Lemma~\ref{lem:SPMC-p4} with repeated applications of Lemma~\ref{lem:mc-induction} the theorem follows. \qed
\end{proof}

\noindent
To prove our next result, we first show a lemma. 

\begin{lemma}\label{lem:cographBoundedColourClass}
    For any red-blue $d$-colouring of a connected $P_4$-free graph, some colour class has size at most $2d$.
\end{lemma}

\begin{proof}
    Let $G$ be a connected $P_4$-free graph with a red-blue $d$-colouring. Then $G$ has at least two vertices. Let $R\subseteq V(G)$ and $B\subseteq V(G)$ denote the (non-empty) sets of red and blue vertices, respectively. A graph is $P_4$-free if and only if it is a cograph. It follows directly from the definition of a cograph that $G$ has a spanning complete bipartite subgraph. Therefore, the vertices of $G$ can be partitioned into a pair of non-empty disjoint sets $S_1,S_2 \neq \emptyset$, which are complete to one another.

    If $S_1$ contains at least $d+1$ red vertices, then $S_2$ must be monochromatic red. Moreover, in that case, $S_1$ contains at most $d$ blue vertices, as each vertex in $S_2$ is adjacent to at most $d$ blue vertices. Hence, $G$ contains at most $d$ blue vertices. 
    Thus, we can assume that $S_1$ and, symmetrically, $S_2$ contain at most $d$ red vertices each. This means that $G$ contains at most $2d$ red vertices.\qed
\end{proof}

\noindent
We are now ready to prove the second main result of this section.

\begin{theorem}\label{t-dc}
For every $d\geq 2$, \dcut\ is polynomial-time solvable on partitioned probe $(P_1+P_4)$-free graphs.
\end{theorem}
\begin{proof}
Let $d\geq 2$.
Let $G=(V,E)$. Let $(G,P,N)$ be a connected partitioned probe $(P_1+P_4)$-free graph, so there exists an edge subset $F\subseteq \binom{N}{2}$ such that $G+F$ is $(P_1+P_4)$-free. 
We show how to find a $d$-cut of $G$ or that none exists.

Suppose that there is a set $Q \subseteq P$ that induces a $P_4$ in $G$.
Then $Q$ dominates $G$, otherwise $G$ contains an induced $P_1+P_4$ of which at most one vertex belongs to $N$ and so $G+F$ also contains an induced $P_1+P_4$.  For every red-blue $d$-colouring of $G$, every vertex $v \in V$ has at most $d$ neighbours in a different colour class.  That is, there are $\mathcal{O}(n^d)$ red-blue $d$-colourings of $N_G[v]$.  As $|Q|=4$, we can consider all $\mathcal{O}(n^{d})$ colourings of $N_G[Q]$, and, as $Q$ dominates $G$, this is all red-blue $d$-colourings of $G$; hence, we can solve the problem in polynomial time.

We may now assume that $G[P]$ is a $P_4$-free graph (cograph) and we consider three cases according to the number of connected components of $G[P]$. As $G$ is connected and $N$ is an independent set of $G$, we find that $P$ dominates $N$, that is, every vertex of $N$ has a neighbour in $P$. This means that $G$ has a red-blue $d$-colouring that colours every vertex of $P$ blue only if $G$ has a red-blue $d$-colouring that colours exactly one vertex of $N$ red and all other vertices of $G$ blue. We can check this in polynomial time. Hence, from now on, we will assume that in every red-blue $d$-colouring of $G$ (if such a colouring exists), there is at least one red vertex and at least one blue vertex in $P$.

     \medskip
    \noindent
    \textbf{Case \thetheorem.1:} $G[P]$ has exactly one connected component.\\ 
\noindent
By Lemma~\ref{lem:cographBoundedColourClass}, we may assume that a set $X$ of at most $2d$ vertices of $P$ is coloured red. We guess $X$.  
As we may assume that $P$ is not monochromatic,
we may assume that $X$ is a proper non-empty subset of $P$.
 We colour every vertex of $P\setminus X$ blue. Then we consider all $\mathcal{O}(n^{2d})$ possible red-blue colourings of the neighbourhood of $X$ in $N$. For each of them, we consider the remaining uncoloured vertices in $N$. As these only have blue neighbours, we can safely colour them blue.  
 It remains to check in polynomial time if the obtained colouring of $G$ is indeed a red-blue $d$-colouring.

     \medskip
    \noindent
    \textbf{Case \thetheorem.2:} $G[P]$ has exactly two connected components.\\ 
Let the two connected components of $G[P]$ have vertex sets $C_1$ and $C_2$.
By Lemma~\ref{lem:cographBoundedColourClass}, we may assume that a set $X_1$ of at most $2d$ vertices of $C_1$ is coloured red and the other vertices of $C_1$ are coloured blue, and that a set $X_2$ of at most $2d$
vertices of $C_2$ are either all coloured red or all coloured blue, while all the other vertices of $C_2$ have the opposite colour. If the vertices of $X_2$ are all coloured red, then we can proceed exactly the same way as in Case~\ref{t-dc}.1. 

It remains to consider those branches where we guess $X_1$ and colour all its vertices red, and we guess $X_2$ and colour all its vertices blue.
In this case, we consider all $\mathcal{O}(n^{4d^2})$ possible red-blue colourings of the neighbourhood of $X_1\cup X_2$ in $N$. For each of them, we colour-process and then consider the uncoloured vertices in $N$. If they only have blue neighbours, we can safely colour them blue.
If they only have red neighbours, we can safely colour them red. 

Suppose that 
there still exist uncoloured vertices in $N$. 
We first consider the case when we have an uncoloured vertex $b\in N$ that in $C_1$ has a neighbour $x$ as well as a non-neighbour $x'$ and that in $C_2$ has a neighbour $y$ as well as a non-neighbour $y'$. 
Since $C_1$ induces a connected graph, we may assume that~$x$ is adjacent to~$x'$ and similarly, we may assume that~$y$ is adjacent to~$y'$.
Hence $\{x',x,b,y,y'\}$ induces a $P_5$ in $G$; see also Figure~\ref{fig-claim-2}. Of the vertices $x',x,b,y,y'$, only $b$ belongs to $N$. Hence, this $P_5$  is also an induced $P_5$ in $G+F$. We now colour all uncoloured vertices in the neighbourhood of $\{x',x,y,y'\}$ in $N$. This yields $\mathcal{O}(n^{4d})$ further branches. 

Suppose that after colour-processing again, we still have an uncoloured vertex~$b'$ in~$N$. 
Note that $b'$ is adjacent to none of $x',x,y,y'$, otherwise it would have been coloured.
If $b'b\notin E(F)$, then $\{b'\}\cup\{x',x,b,y\}$ induces a $P_1+P_4$ in $G+F$. If $b'b\in E(F)$, then $\{y'\}\cup \{x',x,b,b'\}$ induces a $P_1+P_4$ in $G+F$. So, in both cases, we obtain a contradiction. Hence, there are no uncoloured vertices left. We check in polynomial time whether the obtained red-blue colouring of $G$ is a red-blue $d$-colouring. If so, we are done. Otherwise, we discard this branch and consider the next branch.

Now we consider the remaining case, in which every uncoloured vertex in $N$ is either complete or anti-complete to  
either $C_1$ or $C_2$.

First suppose that $N$ contains an uncoloured vertex $b$ that is anti-complete to, say, $C_1$.  
As $b$ is uncoloured, 
$b$ is adjacent to a blue vertex $x$ (as well as to a red vertex). As $b$ is anti-complete to $C_1$, $x$ must belong to $C_2$. Hence, $x$ belongs to~$X_2$. This is a contradiction, as we already coloured all neighbours of $X_2$ in $N$. 

Now suppose that every uncoloured vertex in $N$ is complete to 
either $C_1$ or $C_2$.
Let $b$ be an uncoloured vertex in $N$ that is complete to, say, $C_1$; note that this means that $X_1=\emptyset$.
We recall that every vertex of $C_1 \setminus X_1=C_1$ has the same colour, and also that we colour-processed. Hence, $C_1$ has at most $d$ vertices, as otherwise we would have given $b$ the colour of $C_1$.

We now also colour the neighbourhood of $C_1$ in $N$ in every possible way. This yields $\mathcal{O}(n^{d^2})$ further branches.
Afterwards we colour-process. Suppose we still have an uncoloured vertex $b'$ in $N$. Then $b'$ only has neighbours in $C_2$. As  $b'$
is not adjacent to any vertex in $X_2$, we find that $b'$
only has only red neighbours, and we would have safely coloured $b'$ 
red in a previous step. Hence, there are no uncoloured vertex left in 
$N$.
It remains to check in polynomial time if the obtained red-blue colouring of $G$ is a red-blue $d$-colouring. If so, we are done. Otherwise, we discard this branch and consider the next branch.

As the number of branches is polynomial, and we can process each branch in polynomial time, our algorithm takes polynomial time if this case occurs (the correctness of our algorithm follows from the case description).

\begin{figure}[t]
\centering
\scalebox{.8}{
\begin{tikzpicture}
\coordinate (A) at (-5,0);
\coordinate (B) at (-3,1);
\coordinate (C) at (-2,0);
\coordinate (D) at (0.5,0);
\coordinate (E) at (2.3,1);
\coordinate (F) at (4.3,0);
\coordinate (G) at (-4,-2);
\coordinate (H) at (-2,-1);
\coordinate (I) at (2,-2);
\draw[fill=blue!20!white] (A) rectangle (B)(E) rectangle (F);
\draw[fill=red!20!white] (B) rectangle (C)(D) rectangle (E);
\draw[fill=green!20!white] (G) rectangle (H);
\draw[fill=gray!20!white](H) rectangle (I);
\draw[fill=black]
($(A)!1/2!(B)+(0,-0.2)$) circle [radius=2pt]
($(B)!1/2!(C)+(-0.3,-0.2)$) circle [radius=2pt]
($(D)!1/2!(E)+(0,-0.2)$) circle [radius=2pt]
($(E)!1/2!(F)+(0,-0.2)$) circle [radius=2pt]
($(H)!1/2!(I)+(0,0.2)$) circle [radius=2pt]
($(G)!1/2!(H)+(0,0.2)$) circle [radius=2pt];
\draw ($(B)!1/2!(C)+(-0.3,-0.2)$)--($(A)!1/2!(B)+(0,-0.2)$)--($(G)!1/2!(H)+(0,0.2)$)
($(G)!1/2!(H)+(0,0.2)$)--($(D)!1/2!(E)+(0,-0.2)$)--($(E)!1/2!(F)+(0,-0.2)$);
\draw[dashed] ($(E)!1/2!(F)+(0,-0.2)$)--($(G)!1/2!(H)+(0,0.2)$)--($(B)!1/2!(C)+(-0.3,-0.2)$)
($(D)!1/2!(E)+(0,-0.2)$)--($(H)!1/2!(I)+(0,0.2)$)--($(E)!1/2!(F)+(-0.3,-0.2)$)
($(A)!1/2!(B)+(0,-0.2)$)--($(H)!1/2!(I)+(0,0.2)$)--($(B)!1/2!(C)+(-0.3,-0.2)$);
\node[below] at ($(G)!1/2!(H)+(0,0.2)$) {$b$};
\node[below] at ($(H)!1/2!(I)+(0,0.2)$) {$b'$};
\node[above] at ($(A)!1/2!(B)+(0,-0.2)$) {$x$};
\node[above] at ($(B)!1/2!(C)+(0,-0.2)$) {$x'$};
\node[above] at ($(D)!1/2!(E)+(0,-0.2)$) {$y$};
\node[above] at ($(E)!1/2!(F)+(0,-0.2)$) {$y'$};
\node[above] at ($(A)!1/2!(B)+(-1.7,0)$) {$C_1\setminus X_1$};
\node[above right] at ($(B)!1/2!(C)+(0.6,0)$) {$X_1$};
\node[above left] at ($(D)!1/2!(E)+(-0.9,0)$) {$C_2\setminus X_2$};
\node[above right] at ($(E)!1/2!(F)+(1.1,0)$) {$X_2$};
\node[below right] at ($(B)!1/2!(C)+(0.5,0)$) {$|X_1|\leq 2d$};
\node[below right] at ($(E)!1/2!(F)+(1,0)$) {$|X_2|\leq 4d$};
\node[above left] at ($(G)!1/2!(H)+(-1,0)$) {coloured};
\node[below left] at ($(G)!1/2!(H)+(-1.4,0)$) {$N$};
\node[above right] at ($(H)!1/2!(I)+(2,0)$) {uncoloured};
\node[below right] at ($(H)!1/2!(I)+(2.4,0)$) {$N$};
\end{tikzpicture}}
\caption{Illustration of Case 11.2. Lines represent edges and dashed lines represent non-edges. Note that $x'$ and $y'$ may not belong to $X_1$ and $X_2$, respectively, but this is not relevant (it only matters that they are not adjacent to $b$).}
\label{fig-claim-2}
\vspace*{-0.5cm}
\end{figure}
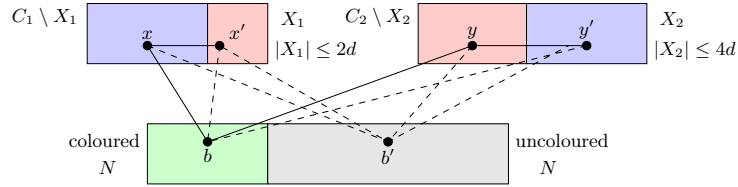

      \medskip
    \noindent
    \textbf{Case \thetheorem.3:} $G[P]$ has at least three connected components.\\
Let the connected components of $G[P]$ have vertex sets $C_1,\ldots,C_r$ for some $r\geq 3$.
We partition $N$ into four types of vertices. Let $v\in N$.
If $v$ is complete to $P$, then we say that $v$ is of {\it type-A}.

Now, suppose that $v$ is not of type-A, but that $v$ has a neighbour in every component of $G[P]$. We say that $v$ is of {\it type-B} and make the following observation.  We know that $v$ is not complete to some set $C_i$, so we can find adjacent vertices $x_i$ and $x_i'$ in $C_i$ such that $x_i$ but not $x_i'$ is adjacent to $v$.  Suppose that $v$ is not complete to another set $C_j$, so there is a vertex $x_j'$ in $C_j$ not adjacent to $v$.   Let $x_k$ be a vertex adjacent to $v$ in the vertex set of a third component. We now have that $\{x_j'\}\cup \{x_i',x_i,v,x_k\}$ forms an induced $P_1+P_4$ in $G$, and consequently in $G+F$, as it only contains one vertex of $N$, namely $v$.  Hence, a type-B vertex is complete to all but one set $C_i$, in which it has least one neighbour.

Now, suppose that $v$ is anti-complete to some $C_h$, so $v$ is neither of type-A nor of type-B. We say that $v$ is of {\it type-C} if $v$ has a neighbour in at least two other components $C_i$ and $C_j$. Suppose that $v$ has a  non-neighbour in either $C_i$ or $C_j$, say in $C_i$.
Let $x_h\in C_h$, $x_i,x_i'\in C_i$ such that $x_ix_i'$ is an edge with only $x_i$ adjacent to $v$, and let $x_j\in C_j$ be adjacent to $v$. Now, $\{x_h\}\cup \{x_i',x_i,v,x_j\}$ induces a $P_1+P_4$ in $G$, and also in $G+F$ as it has only one vertex from $N$, a  contradiction as $G+F$ is $(P_1+P_4)$-free.
Hence, $v$ is complete to every $C_j$ in which it has a neighbour.
So, a type-C vertex of $N$ is complete to at least two and at most $r-1$ sets in $\{C_1,\ldots,C_r\}$ and anti-complete to all other sets in $\{C_1,\ldots,C_r\}$.

Finally, if $v$ is neither type-A nor type-B nor type-C, then $v$ is of type-D. As $G$ is connected, a type-D vertex $v$ has at least one neighbour in one set $C_i$ and is anti-complete to every other set in $\{C_1,\ldots,C_r\}$. See Figure~\ref{fig-types} for an illustration.

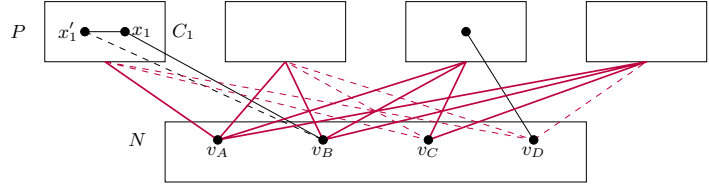
\begin{figure}[t]
\centering
\scalebox{.8}{
\begin{tikzpicture}
\coordinate (A) at (-6,0);
\coordinate (B) at (-4,1);
\coordinate (C) at (-3,0);
\coordinate (D) at (-1,1);
\coordinate (E) at (0,0);
\coordinate (F) at (2,1);
\coordinate (G) at (3,0);
\coordinate (H) at (5,1);
\coordinate (I) at (-4,-2);
\coordinate (L) at (3,-1);
\coordinate (M) at ($(I)+(0,0.7)$);
\coordinate (N) at ($(L)+(0,-0.3)$);
\coordinate (S) at ($(A)+(0,0.5)$);
\coordinate (T) at ($(B)+(0,-0.5)$);
\draw[fill=white]
(A)rectangle(B)(C)rectangle(D)(E)rectangle(F)(G)rectangle(H)(I)rectangle(L);
\draw
($(M)!3/8!(N)$)--($(S)!2/3!(T)$)--($(S)!1/3!(T)$)
($(E)!1/2!(F)$)--($(M)!7/8!(N)$);
\draw[dashed] ($(S)!1/3!(T)$)--($(M)!3/8!(N)$);
\draw[thick,color=purple]
($(A)!1/2!(B)+(0,-0.5)$)--($(M)!1/8!(N)$)--($(C)!1/2!(D)+(0,-0.5)$)
($(E)!1/2!(F)+(0,-0.5)$)--($(M)!1/8!(N)$)--($(G)!1/2!(H)+(0,-0.5)$)
($(G)!1/2!(H)+(0,-0.5)$)--($(M)!3/8!(N)$)--($(C)!1/2!(D)+(0,-0.5)$)
($(E)!1/2!(F)+(0,-0.5)$)--($(M)!3/8!(N)$)
($(E)!1/2!(F)+(0,-0.5)$)--($(M)!5/8!(N)$)--($(G)!1/2!(H)+(0,-0.5)$);
\draw[dashed,color=purple]
($(A)!1/2!(B)+(0,-0.5)$)--($(M)!5/8!(N)$)--($(C)!1/2!(D)+(0,-0.5)$)
($(A)!1/2!(B)+(0,-0.5)$)--($(M)!7/8!(N)$)--($(C)!1/2!(D)+(0,-0.5)$)
($(M)!7/8!(N)$)--($(G)!1/2!(H)+(0,-0.5)$);
\draw[fill=black]
\foreach \i in {1,3,5,7} {($(M)!\i/8!(N)$) circle [radius=2pt]}
($(E)!1/2!(F)$) circle [radius=2pt]
\foreach \i in {1,2} {($(S)!\i/3!(T)$) circle [radius=2pt]};
\foreach \i/\j in {1/A, 3/B, 5/C, 7/D}{
\node[below] at ($(M)!\i/8!(N)$) {$v_\j$};}
\node[left] at ($(S)+(-0.2,0)$) {$P$};
\node[left] at ($(M)+(-0.2,0)$) {$N$};
\node[right] at (T) {$C_1$};
\node[left] at ($(S)!1/3!(T)$) {$x_1'$};
\node[right] at ($(S)!2/3!(T)$) {$x_1$};
\end{tikzpicture}}
\caption{Illustration of the types described in Case \thetheorem.3, specified by vertex label. A (dashed) purple line between a vertex $v\in N$ and some $C_i$ represents that $v$ is (anti-)complete to $C_i$. A black (dashed) edge between a vertex $v\in N$ and a vertex in $P$ represents a (non)-edge.}
\label{fig-types}
\vspace*{-0.5cm}
\end{figure}

\noindent
We distinguish between the following three cases:

\medskip
\noindent
{\bf Case 11.3.1} $N$ has a type-A vertex.\\
Let $v\in N$ be of type-A. We colour $v$ blue and we consider all $\mathcal{O}(n^d)$ possible red-blue colourings of its neighbourhood $N(v)=P$. We then find a set $Q\subseteq P$ of at most $d$ vertices in $P$ that are coloured red. As we already checked whether $G$ has a red-blue $d$-colouring in which $P$ is monochromatic, we may assume that $Q$ is a proper non-empty subset of $P$. We now consider all possible  $\mathcal{O}(n^{d^2})$ possible red-blue colourings of the neighbourhood of $Q$ in $N$. For each one of them, the uncoloured vertices in $N$ only have blue neighbours and form an independent set, we can safely colour them blue. Note that at least one vertex is red and at least one is blue. It remains to check whether the obtained colouring of $G$ is a red-blue $d$-colouring. If so, we are done, and otherwise we discard the branch.
Hence, as the number of branches is polynomial, and we can process each branch in polynomial time, this case takes polynomial time.

\medskip
\noindent
{\bf Case 11.3.2} $N$ has no type-A vertices, but $N$ has a type-B vertex.\\
Let $v \in N$ be of type-B.
We assume without loss of generality that $v$ has a neighbour and a non-neighbour in $C_1$ and is complete to $C_2,\ldots,C_r$.
We colour~$v$ blue, and we consider all $\mathcal{O}(n^d)$ options of choosing a set $X_v$ of at most $d$ red neighbours of $v$ in $P$. We colour all other neighbours of $v$ in $P$ blue. Hence, afterwards part of $C_1$ is coloured, while all vertices of every $C_i$ with $i\in \{2,\ldots,r\}$ are coloured.
Moreover, by Lemma~\ref{lem:cographBoundedColourClass}, we find that any red-blue $d$-colouring of $G$ (if one exists) either colours at most $2d$ vertices of $C_1$ red, or else it colours at most $2d$ vertices of $C_1$ blue. Hence, we also consider all $\mathcal{O}(n^{2d})$ options for choosing a set $X\subseteq C_1$ of size at most~$2d$ that consists of these vertices. We colour the vertices of $C_1\setminus X$ with the opposite colour. We also colour the neighbourhood of $X_v \cup X$ in $N$ in every possible way. This yields another $\mathcal{O}(n^{3d^2})$ branches. For each branch, we colour-process. If afterwards we find an uncoloured vertex in $N$ whose neighbourhood in $P$ is monochromatic, then we give it the unique colour of the vertices in its neighbourhood in $P$.

Suppose that afterwards there still exist uncoloured vertices in $N$. Let $b$ be an arbitrary uncoloured vertex in $N$
and note that
$b$ is adjacent to both a blue vertex $x$ and a red vertex $y$ in $P$. As we coloured the neighbourhood of $X_v\cup X$, neither $x$ nor $y$ belongs to $X_v\cup X$. 
As $X_v$ contains all red vertices in $N(v)$, it follows that $y$ belongs to $C_1 \setminus N(v)$, and thus $y\in C_1\setminus X$. Consequently, all vertices of $X$ are blue. As $x$, which is blue, is not in $X$, and all vertices in $C_1\setminus X$ are red, $x$ belongs to some set $C_i$ with $i\geq 2$. 
Hence, all uncoloured vertices have a neighbour in $C_1$ and a neighbour in at least one other $C_i$. As $N$ has no type-A vertices, every uncoloured vertex in $N$ is of type-B or of type-C. 

First, suppose that $N$ has an uncoloured vertex $b$ that is of type-B. By definition, $b$ is complete to $r-1$ sets in $\{C_1,\ldots,C_r\}$. Let $Y$ be the union of these $r-1$ sets. We recall that $b$ is uncoloured and that all vertices in every $C_i$ are coloured. Hence, $Y$ must contain at most $2d$ vertices in total, otherwise we would have given $b$ a colour during colour-processing. We now consider all $\mathcal{O}(n^{2d^2})$ possible red-blue colourings of the neighbourhood of $Y$ in $N$. Afterwards, there are no uncoloured vertices left, as every uncoloured vertex was of type-B or of type-C and thus must have a neighbour in $Y$.
It remains to check whether the obtained colouring of $G$ is a red-blue $d$-colouring. If so, we are done, and otherwise we discard the branch.
 
So we can now assume
that all uncoloured vertices are of type-C. Let $b$ again denote an uncoloured vertex in $N$. From the definition of type-C, it follows that $b$ is either complete or anti-complete to every set $C_i$. Recall that all uncoloured vertices of $N$, and thus $b$, have a neighbour in $C_1$.
Hence, $b$ is complete to $C_1$. Note that $b$ is anti-complete to $X$, as otherwise $b$ would have been coloured. This means that $X=\emptyset$, and thus every vertex of $C_1$ is coloured red. As $b$ is uncoloured and we have colour-processed, this means that $C_1$ has size at most $d$. Hence, we can consider  all $\mathcal{O}(n^{d^2})$ possible red-blue colourings of the neighbourhood of $C_1$ in $N$. 
Afterwards, all uncoloured vertices in $N$ have received a colour (as they were all complete to $C_1$).  It remains to check whether the obtained colouring of $G$ is a red-blue $d$-colouring. If so, we are done, otherwise we discard the branch.

As the number of branches is polynomial, and we can process each branch in polynomial time, our algorithm takes polynomial time if this case occurs.

\medskip
\noindent
{\bf Case 11.3.3} $N$ has no type-A and no type-B vertices.\\
As $G$ is connected and $r\geq 3$, $G$ must have vertices that are adjacent to at least two sets $C_i$.
As $G$ has no vertices of type-A or type-B, all of these vertices must be of type-C.
Let $u$ and $v$ be two type-C vertices, such that the following holds:
\begin{enumerate}[(i)]
\item $u$ is complete to some $C_h$, but anti-complete to $C_i$;
\item $v$ is anti-complete to $C_h$, but complete to $C_i$; and
\item $u$ and $v$ are both complete to some $C_j$.
\end{enumerate}
In this case, we claim that $\{u,v\}$ is a {\it $P$-dominating pair}, that is, every $C_i$ is complete to a least one of $u$, $v$; see Figure~\ref{fig-p-dominating-pair}. 
For a contradiction, suppose that neither $u$ nor $v$ is complete to some $C_k$ (so both are anti-complete to $C_k$ as they are of type-C). Let $z\in C_k$. If $uv\notin E(F)$, then a vertex of $C_h$, $u$, a vertex of $C_j$ and $v$ form, together with $z$, an induced $P_1+P_4$ in $G+F$. If $uv\in E(F)$, then a vertex of $C_h$, $u$, $v$ and a vertex of $C_i$, together with $z$, form an induced $P_1+P_4$ in $G+F$. So, in both cases, we derive a contradiction. 

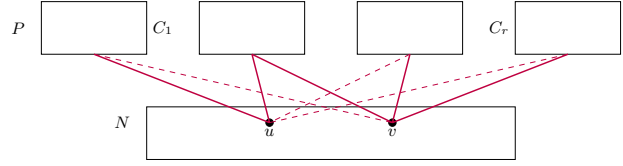
\begin{figure}[t]
\centering
\scalebox{.7}{
\begin{tikzpicture}
\coordinate (A) at (-6,0);
\coordinate (B) at (-4,1);
\coordinate (C) at (-3,0);
\coordinate (D) at (-1,1);
\coordinate (E) at (0,0);
\coordinate (F) at (2,1);
\coordinate (G) at (3,0);
\coordinate (H) at (5,1);
\coordinate (I) at (-4,-2);
\coordinate (L) at (3,-1);
\coordinate (M) at ($(I)+(0,0.7)$);
\coordinate (N) at ($(L)+(0,-0.3)$);
\coordinate (S) at ($(A)+(0,0.5)$);
\coordinate (T) at ($(B)+(0,-0.5)$);
\draw[fill=white]
(A)rectangle(B)(C)rectangle(D)(E)rectangle(F)(G)rectangle(H)(I)rectangle(L);
\draw[fill=black]
\foreach \i in {2,4} {($(M)!\i/6!(N)$) circle [radius=2pt]};
\draw[thick,color=purple]
($(A)+(1,0)$)--($(M)!2/6!(N)$)
($(C)+(1,0)$)--($(M)!2/6!(N)$)
($(C)+(1,0)$)--($(M)!4/6!(N)$)
($(E)+(1,0)$)--($(M)!4/6!(N)$)
($(G)+(1,0)$)--($(M)!4/6!(N)$);
\draw[dashed,color=purple]
($(E)+(1,0)$)--($(M)!2/6!(N)$)
($(G)+(1,0)$)--($(M)!2/6!(N)$)
($(A)+(1,0)$)--($(M)!4/6!(N)$);
\foreach \i/\j in {2/u,4/v}{\node[below] at ($(M)!\i/6!(N)$) {$\j$};}
\node[left] at ($(S)+(-0.2,0)$) {$P$};
\node[left] at ($(M)+(-0.2,0)$) {$N$};
\node[right] at ($(T)$) {$C_1$};
\node[left] at ($(G)+(0,0.5)$) {$C_r$};
\end{tikzpicture}}
\caption{A $P$-dominating pair $\{u,v\}$.}
\label{fig-p-dominating-pair}
\vspace*{-0.5cm}
\end{figure}

We now continue with the description of our algorithm. We choose a type-C vertex $v\in N$ such that $v$ is complete to a maximum number of sets $C_i$ over all type-C vertices of $N$. We say that $v$ is of {\it maximum} type-C. 
From the definition of type-C, it follows that $v$ is anti-complete to at least one set in $\{C_1,\ldots,C_r\}$. Let $C_h$ be such a set. As $G$ is connected, $G$ contains a path from $v$ to the vertices in $C_h$.
Hence, without loss of generality, there exists a type-C vertex $u\in N$ that is complete to $C_h$ and to at least one other $C_j$ to which $v$ is also complete. 
As $v$ is of maximum type-C and $v$ is not complete to $C_h$, to which $u$ is complete, 
we find that $u$ is also anti-complete to some set $C_i$ to which $v$ is complete. Hence, $\{u,v\}$ satisfies (i)--(iii), so $\{u,v\}$ is a $P$-dominating pair.

As $\{u,v\}$ is a $P$-dominating pair, we can colour $P$ as follows. First suppose $u$ and $v$ are coloured alike, say both are coloured blue. We proceed in exactly the same way as in Case~\thetheorem.1. As $u$ and $v$ can each have at most $d$ red neighbours, 
we may assume that a set $X$ of at most $2d$ vertices of $P$ is coloured red. We guess $X$. As we already checked whether $G$ has a red-blue $d$-colouring in which $P$ is monochromatic, we may assume that $X$ is a proper non-empty subset of $P$. We colour every vertex of $P\setminus X$ red. Afterwards, we consider all $\mathcal{O}(n^{2d})$ possible red-blue colourings of the neighbourhood of $X$ in $N$. For each one of them, we consider the uncoloured vertices in $N$. As these only have blue neighbours, we can safely colour them blue. Note that at least one vertex is red and at least one vertex is blue. It remains to check in polynomial time whether the obtained colouring of $G$ is indeed a red-blue $d$-colouring.

Now suppose $u$ is coloured red and $v$ is coloured blue. We guess a set $X_u$ of at most $d$ blue neighbours of $u$ in $P$. We colour all other neighbours of $u$ in $P$ red. Similarly, we guess a set $X_v$ of at most $d$ blue neighbours of $v$ in $P$. We colour all other neighbours of $v$ in $P$ red. 
This gives us $\mathcal{O}(n^{2d})$ branches. We note that every vertex of $P$ has been coloured, as $\{u,v\}$ is a $P$-dominating pair. We now colour the neighbourhood of $X_u\cup X_v$ in $N$ in every possible way. This gives us a further $\mathcal{O}(n^{2d^2})$ branches.

We let $C^*_1,\ldots, C^*_q$ be those sets $C_i$ that contain a vertex of $X_u\cup X_v$. 
We let $C^u_1,\ldots,C^u_s$ be those sets $C_i$ that contain no vertex of $X_u\cup X_v$ and are coloured red. 
We let $C^v_1,\ldots,C^v_t$ be those sets $C_i$ that contain no vertex of $X_u\cup X_v$ and are coloured blue. Note that every $C_i$ belongs to one of these three families of sets, but some of these families might be empty.

As $v$ is coloured blue, its neighbours not in $X_v$ are blue. Hence, $v$ is anti-complete to every $C^u_i$. Similarly, $u$ is anti-complete to every $C^v_j$. 
Recall that $\{u,v\}$ is a $P$-dominating pair and that each of $u,v$, being of type-C, is either complete or anti-complete to a set $C_i$.
Consequently, $u$ is complete to every $C^u_i$ and $v$ is complete to every $C^v_i$.

If $C^u_1$ exists, then we select an arbitrary vertex $x_1\in C^u_1$, and we colour the neighbourhood of $x_1$ in $N$ in every possible way. This leads to $\mathcal{O}(n^d)$ additional branches.
We do the same if $C^v_1$ exists, leading to another $\mathcal{O}(n^d)$ branches.
We now colour-process. Afterwards, we colour any vertex in $N$ with monochromatic neighbourhood in $P$ with the unique colour of its neighbours in $P$.

We claim that every vertex of type-D has now been coloured. For a contradiction, suppose $z\in N$ is of type-D and has no colour yet. We recall that all $C_i^u$ and $C_j^v$ are monochromatic and that $z$, being type-D, only has neighbours in exactly one $C_i$.
This means that $z$ must have both a blue neighbour and a red neighbour in some $C_i^*$. However, one of these two neighbours of $z$ belongs to $X_u\cup X_v$, meaning $z$ would have been coloured.

Hence, the only vertices of $G$ that are possibly still uncoloured are type-C vertices in $N$.
Let $b\in N$ be an uncoloured vertex of type-C.  
This means that $b$ has both a red neighbour and a blue neighbour in $P$. 
By definition, $b$ is either complete or anti-complete to $C_i$ for every $i\in \{1,\ldots,r\}$.
This means that $b$ is anti-complete to every $C_h^*$, as we coloured the neighbourhood of $C_h^*\cap (X_u\cup X_v)$, which is a non-empty set by definition. 
Hence, the red neighbour of $b$ must be some $z^u\in C^u_i$, and the blue neighbour of $b$ must be some $z^v\in C^v_j$. From the above, we find that $b$ is anti-complete to $C_1^u$ and $C_1^v$ (as we coloured the neighbourhood of one of the vertices in them).
Hence, we have $i\geq 2$ and $j\geq 2$. 
Let $y^u\in C_1^u$ and $y^v\in C_1^v$. 
If $bv\notin F$, then $\{y^u\}\cup \{z^u,b,z^v,v\}$ induces a $P_1+P_4$ in $G+F$.
If $bv\in F$, then $\{y^u\}\cup \{z^u,b,v,y^v\}$ induces a $P_1+P_4$ in $G+F$.
See Figure~\ref{fig-c-types} for an illustration.
In both cases, we obtain a contradiction.

\begin{figure}[t]
\centering
\begin{minipage}{0.5\textwidth}
\begin{tikzpicture}[scale=0.6]
\coordinate (A) at (-6,0);
\coordinate (B) at (-4.5,1);
\coordinate (C) at (-3.5,0);
\coordinate (D) at (-2,1);
\coordinate (G) at (-1,0);
\coordinate (H) at (0.5,1);
\coordinate (X) at (1.5,0);
\coordinate (Y) at (3,1);
\coordinate (I) at (-4.5,-2);
\coordinate (L) at (1.5,-1);
\coordinate (M) at ($(I)+(0,0.7)$);
\coordinate (N) at ($(L)+(0,-0.3)$);
\coordinate (S) at ($(A)+(0,0.5)$);
\coordinate (T) at ($(B)+(0,-0.5)$);
\draw[fill=red!20!white](A)rectangle(B)(C)rectangle(D);
\draw[fill=blue!20!white](G)rectangle(H)(X)rectangle(Y);
\draw[fill=white] (I)rectangle(L);
\draw
($(A)!1/2!(B)+(0,-0.2)$)--($(M)!1/6!(N)$)
($(C)!1/2!(D)+(0,-0.2)$)--($(M)!1/6!(N)$)
($(G)!1/2!(H)+(0,-0.2)$)--($(M)!3/6!(N)$);
\draw[dashed]
($(A)!1/2!(B)+(0,-0.2)$)--($(M)!3/6!(N)$)
($(C)!1/2!(D)+(0,-0.2)$)--($(M)!3/6!(N)$)
($(G)!1/2!(H)+(0,-0.2)$)--($(M)!1/6!(N)$)
($(X)!1/2!(Y)+(0,-0.2)$)--($(M)!1/6!(N)$)
($(A)!1/2!(B)+(0,-0.2)$)--($(M)!5/6!(N)$)
($(G)!1/2!(H)+(0,-0.2)$)--($(M)!5/6!(N)$);
\draw[dashed,blue!50!white,very thick]
($(M)!5/6!(N)$)--($(M)!3/6!(N)$);
\draw[very thick]
($(C)!1/2!(D)+(0,-0.2)$)--($(M)!5/6!(N)$)--($(X)!1/2!(Y)+(0,-0.2)$)--($(M)!3/6!(N)$);
\draw[fill=red] ($(A)!1/2!(B)+(0,-0.2)$) circle[radius=4pt]
($(C)!1/2!(D)+(0,-0.2)$) circle[radius=4pt];
\draw[fill=blue] ($(G)!1/2!(H)+(0,-0.2)$) circle[radius=2pt]
($(X)!1/2!(Y)+(0,-0.2)$) circle[radius=4pt];
\node[above] at ($(A)!1/2!(B)+(0,-0.2)$) {$y^u$};
\node[above] at ($(C)!1/2!(D)+(0,-0.2)$) {$z^u$};
\node[above] at ($(G)!1/2!(H)+(0,-0.2)$) {$y^v$};
\node[above] at ($(X)!1/2!(Y)+(0,-0.2)$) {$z^v$};
\foreach \i/\j/\k in {B/1/u,D/i/u,H/1/v,Y/j/v} {
\node[right] at ($(\i)+(0,-0.5)$) {$C_{\j}^{\k}$};}
\foreach \i/\j/\k in {1/red/2,3/blue/4,5/white/4} {\draw[fill=\j]($(M)!\i/6!(N)$) circle [radius=\k pt];}
\foreach \i/\j in {1/u,3/v,5/b}{\node[below] at ($(M)!\i/6!(N)$) {$\j$};}
\node[left] at ($(S)+(-0.2,0)$) {$P$};
\node[left] at ($(M)+(-0.2,0)$) {$N$};
\end{tikzpicture}
\end{minipage}%
\begin{minipage}{0.5\textwidth}
\begin{tikzpicture}[scale=0.6]
\coordinate (A) at (-6,0);
\coordinate (B) at (-4.5,1);
\coordinate (C) at (-3.5,0);
\coordinate (D) at (-2,1);
\coordinate (G) at (-1,0);
\coordinate (H) at (0.5,1);
\coordinate (X) at (1.5,0);
\coordinate (Y) at (3,1);
\coordinate (I) at (-4.5,-2);
\coordinate (L) at (1.5,-1);
\coordinate (M) at ($(I)+(0,0.7)$);
\coordinate (N) at ($(L)+(0,-0.3)$);
\coordinate (S) at ($(A)+(0,0.5)$);
\coordinate (T) at ($(B)+(0,-0.5)$);
\draw[fill=red!20!white](A)rectangle(B)(C)rectangle(D);
\draw[fill=blue!20!white](G)rectangle(H)(X)rectangle(Y);
\draw[fill=white] (I)rectangle(L);
\draw
($(A)!1/2!(B)+(0,-0.2)$)--($(M)!1/6!(N)$)
($(C)!1/2!(D)+(0,-0.2)$)--($(M)!1/6!(N)$)
($(C)!1/2!(D)+(0,-0.2)$)--($(M)!5/6!(N)$)
($(X)!1/2!(Y)+(0,-0.2)$)--($(M)!5/6!(N)$);
\draw[dashed]
($(A)!1/2!(B)+(0,-0.2)$)--($(M)!3/6!(N)$)
($(C)!1/2!(D)+(0,-0.2)$)--($(M)!3/6!(N)$)
($(G)!1/2!(H)+(0,-0.2)$)--($(M)!1/6!(N)$)
($(X)!1/2!(Y)+(0,-0.2)$)--($(M)!1/6!(N)$)
($(A)!1/2!(B)+(0,-0.2)$)--($(M)!5/6!(N)$)
($(G)!1/2!(H)+(0,-0.2)$)--($(M)!5/6!(N)$);
\draw[blue!50!white,very thick]
($(M)!5/6!(N)$)--($(M)!3/6!(N)$);
\draw[very thick]
($(C)!1/2!(D)+(0,-0.2)$)--($(M)!5/6!(N)$)($(M)!3/6!(N)$)--($(G)!1/2!(H)+(0,-0.2)$);
\draw[fill=red] ($(A)!1/2!(B)+(0,-0.2)$) circle[radius=4pt]
($(C)!1/2!(D)+(0,-0.2)$) circle[radius=4pt];
\draw[fill=blue] ($(G)!1/2!(H)+(0,-0.2)$) circle[radius=4pt]
($(X)!1/2!(Y)+(0,-0.2)$) circle[radius=2pt];
\node[above] at ($(A)!1/2!(B)+(0,-0.2)$) {$y^u$};
\node[above] at ($(C)!1/2!(D)+(0,-0.2)$) {$z^u$};
\node[above] at ($(G)!1/2!(H)+(0,-0.2)$) {$y^v$};
\node[above] at ($(X)!1/2!(Y)+(0,-0.2)$) {$z^v$};
\foreach \i/\j/\k in {1/red/2,3/blue/4,5/white/4} {\draw[fill=\j]($(M)!\i/6!(N)$) circle [radius=\k pt];}
\foreach \i/\j in {1/u,3/v,5/b}{\node[below] at ($(M)!\i/6!(N)$) {$\j$};}
\node[left] at ($(S)+(-1,0)$) {};
\end{tikzpicture}
\end{minipage}
\caption{Illustration of Case~11.3.3, where there is a $P$-dominating pair $\{u,v\}$ made of vertices of type-C and another type-C vertex $b$ that has not been coloured yet. The graph $G+F$ contains an induced $P_1+P_4$ (highlighted by large vertices and thick edges) regardless of whether there is an edge between $v$ and $b$ or not.}
\label{fig-c-types}
\end{figure}
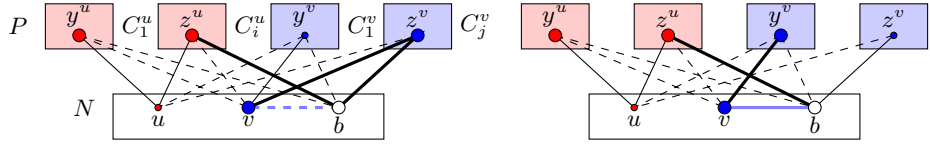

From the above, we conclude that there are no uncoloured vertices, and we have obtained a red-blue colouring of $G$. It now remains to check in polynomial time whether this is a red-blue $d$-colouring. If so, then we are done. Otherwise, we discard this branch and move on to the next branch. As the number of branches is polynomial, and we can process each branch in polynomial time, our algorithm runs in polynomial time (its correctness follows from its description). \qed
\end{proof}

\section{NP-Completeness Results}\label{s-np}

To finish the proof of Theorem~\ref{t-dicho}, we must show that
that $1$-{\sc Cut} and {\sc Perfect Matching Cut} are \NP-complete on probe $2P_2$-free graphs and probe $K_{1,3}$-free graphs and that for $d\geq 2$, \dcut\ is \NP-complete on probe $2P_2$-free graphs and probe $4P_1$-free graphs. All other \NP-completeness results follow directly from the corresponding \NP-completeness results in Theorems~\ref{t-1} and~\ref{t-2}. We recall that the \NP-completeness of $2$-{\sc Cut} for $K_{1,3}$-free graphs was recently shown~\cite{AELPS}.
As we shall see below, we often prove \NP-completeness for even more restricted classes of probe graphs.

We start with the following result, which we recall is in contrast to the polynomial-time result of $1$-{\sc Cut} for $K_{1,3}$-free graphs due to Bonsma~\cite{Bo09}. 
Its proof is based on an observation of Moshi~\cite{Mo89}.

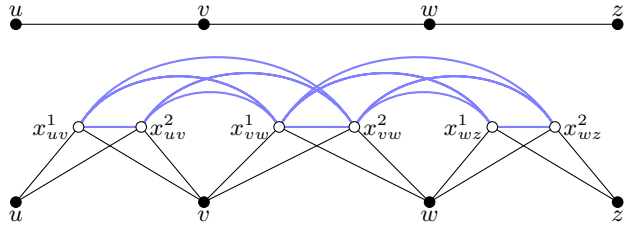
\begin{figure}[t]
\centering
\begin{minipage}{0.9\textwidth}
\centering
\begin{tikzpicture}
\coordinate (Z) at (4,0);
\coordinate (W) at (1.5,0);
\coordinate (V) at (-1.5,0);
\coordinate (U) at (-4,0);
\draw (U)--(V)--(W)--(Z);
\draw[fill=black] 
(U) circle [radius=2pt]
(V) circle [radius=2pt]
(W) circle [radius=2pt]
(Z) circle [radius=2pt];
\foreach \i/\j in {U/u, V/v, W/w, Z/z}{
\node[above] at (\i) {$\j$};}
\end{tikzpicture}
\end{minipage}
\begin{minipage}{0.9\textwidth}
\centering
\begin{tikzpicture}
\coordinate (Z) at (4,0);
\coordinate (W) at (1.5,0);
\coordinate (V) at (-1.5,0);
\coordinate (U) at (-4,0);
\draw
($(U)!1/3!(V)+(0,1)$)--(U)--($(U)!2/3!(V)+(0,1)$)
($(U)!1/3!(V)+(0,1)$)--(V)--($(U)!2/3!(V)+(0,1)$)
($(V)!1/3!(W)+(0,1)$)--(V)--($(V)!2/3!(W)+(0,1)$)
($(V)!1/3!(W)+(0,1)$)--(W)--($(V)!2/3!(W)+(0,1)$)
($(W)!1/3!(Z)+(0,1)$)--(W)--($(W)!2/3!(Z)+(0,1)$)
($(W)!1/3!(Z)+(0,1)$)--(Z)--($(W)!2/3!(Z)+(0,1)$);
\draw[color=blue!50!white,thick]
($(U)!1/3!(V)+(0,1)$)--($(U)!2/3!(V)+(0,1)$)
($(V)!1/3!(W)+(0,1)$)--($(V)!2/3!(W)+(0,1)$)
($(W)!1/3!(Z)+(0,1)$)--($(W)!2/3!(Z)+(0,1)$)
\foreach \i in {1,2} {
\foreach \j in {1,2} {
($(U)!\i/3!(V)+(0,1)$)to[out=60,in=120]($(V)!\i/3!(W)+(0,1)$)to[out=60,in=120]($(W)!\i/3!(Z)+(0,1)$)}}
($(U)!1/3!(V)+(0,1)$)to[out=60,in=120]($(V)!2/3!(W)+(0,1)$)
($(U)!2/3!(V)+(0,1)$)to[out=60,in=120]($(V)!1/3!(W)+(0,1)$)
($(V)!1/3!(W)+(0,1)$)to[out=60,in=120]($(W)!2/3!(Z)+(0,1)$)
($(V)!2/3!(W)+(0,1)$)to[out=60,in=120]($(W)!1/3!(Z)+(0,1)$);
\draw[fill=black] 
(U) circle [radius=2pt]
(V) circle [radius=2pt]
(W) circle [radius=2pt]
(Z) circle [radius=2pt];
\draw[fill=white]
\foreach \i in {1,2} {
($(U)!\i/3!(V)+(0,1)$) circle [radius=2pt]
($(V)!\i/3!(W)+(0,1)$) circle [radius=2pt]
($(W)!\i/3!(Z)+(0,1)$) circle [radius=2pt]};
\foreach \i/\j in {U/u, V/v, W/w, Z/z}{
\node[below] at (\i) {$\j$};}
\node[left] at ($(U)!1/3!(V)+(0,1)$) {$x_{uv}^1$};
\node[left] at ($(V)!1/3!(W)+(0,1)$) {$x_{vw}^1$};
\node[left] at ($(W)!1/3!(Z)+(0,1)$) {$x_{wz}^1$};
\node[right] at ($(U)!2/3!(V)+(0,1)$) {$x_{uv}^2$};
\node[right] at ($(V)!2/3!(W)+(0,1)$) {$x_{vw}^2$};
\node[right] at ($(W)!2/3!(Z)+(0,1)$) {$x_{wz}^2$};
\end{tikzpicture}
\end{minipage}
\caption{An example of a graph $G$ from the proof of Theorem~\ref{t-claw} with edges $uv$, $vw$ and $wz$, together with the graph $G'+F$, where the edges of $F$ are blue.}\label{fig-claw}
\end{figure}

\begin{theorem}\label{t-claw}
$1$-{\sc Cut} is \NP-complete on probe $K_{1,3}$-free graphs.
\end{theorem}

\begin{proof}
    We reduce from $1$-{\sc Cut}, which we recall is \NP-complete~\cite{Ch84}. From a connected graph $G=(V,E)$ we construct a connected graph $G'$ by replacing each edge $uv$ with 
     two new vertices $x_{uv}^1$ and $x_{uv}^2$, which we call {\it intermediate}, and edges $ux_{uv}^1$, $ux_{uv}^2$, $vx_{uv}^1$ and $vx_{uv}^2$. See Figure~\ref{fig-claw} for an example.

    We claim that $G'$ is probe $K_{1,3}$-free. To see this, let $N$ be the set of all intermediate vertices. Note that $N$ is an independent set in $G'$. Let $F$ consist of all edges between any two intermediate vertices that have at least one common neighbour in $V$;
    see also Figure~\ref{fig-claw}.
    
     We claim that $G'+F$ is $K_{1,3}$-free.
 For a contradiction, suppose that $G'+F$ has an induced
claw
  say with center~$z$.
By construction, the neighbourhood in~$G'+F$ of every $u \in V$ is a clique.
Hence, $z=x_{uv}^i$ for some $uv\in E$ and $i\in \{1,2\}$. Furthermore, all neighbours of $z$ in $V(G'+F)$ can be partitioned into a pair of cliques. The first consists of $u$ and $x_{uw}^1,x_{uw}^2$ for every $w$ such that $uw \in E$ and the second consists of $v$ and $x_{wv}^1,x_{wv}^2$ for every $w$ such that $wv \in E$. This contradicts our assumption that $z$ is the centre of an induced claw.

Moshi~\cite{Mo89} showed that a graph has a $1$-cut if and only if the graph obtained from it by replacing an edge $uv$ with 
     two new vertices $x_{uv}^1$ and $x_{uv}^2$, and edges $ux_{uv}^1$, $ux_{uv}^2$, $vx_{uv}^1$ and $vx_{uv}^2$ has a matching cut.
     Hence, $G$ has a $1$-cut if and only if $G'$ has a $1$-cut, which means the theorem is proven. \qed
\end{proof}

\noindent
The diamond $\overline{2P_1+P_2}$ is obtained from taking the complement of $2P_1+P_2$, or equivalently, from the $K_4$ after removing an edge.
A graph is {\it subcubic} if it has maximum degree at most~$3$.
The proof of our next result is similar to a corresponding result for {\sc Vertex Cover} from \cite{BOPPRL25} except that we reduce from a different known \NP-complete problem. 

 \begin{figure}[t]
\centering
\begin{minipage}{0.3\textwidth}
\centering
\begin{tikzpicture}[scale=0.7]
\coordinate (A) at (-1.5,2);
\coordinate (B) at (1.5,2);
\coordinate (C) at (-1.5,-2);
\coordinate (D) at (1.5,-2);
\coordinate (E) at (0,0.8);
\coordinate (F) at (0,-0.8);
\draw (A)--(B)--(D)--(C)--(A)--(E)--(B)(E)--(F)(C)--(F)--(D);
\draw[fill=black] 
(A) circle [radius=3pt]
(B) circle [radius=3pt]
(C) circle [radius=3pt]
(D) circle [radius=3pt]
(E) circle [radius=3pt]
(F) circle [radius=3pt];
\end{tikzpicture}
\end{minipage}%
\begin{minipage}{0.3\textwidth}
\centering
\begin{tikzpicture}[scale=0.7]
\coordinate (A) at (-1.5,2);
\coordinate (B) at (1.5,2);
\coordinate (C) at (-1.5,-2);
\coordinate (D) at (1.5,-2);
\coordinate (E) at (0,0.8);
\coordinate (F) at (0,-0.8);
\draw (A)--(B)--(D)--(C)--(A)--(E)--(B)(E)--(F)(C)--(F)--(D);
\draw[fill=black] 
(A) circle [radius=3pt]
(B) circle [radius=3pt]
(C) circle [radius=3pt]
(D) circle [radius=3pt]
(E) circle [radius=3pt]
(F) circle [radius=3pt];
\draw[fill=white]
\foreach \i in {1,2,3,4} {
($(A)!\i/5!(B)$) circle [radius=2pt]
($(A)!\i/5!(C)$) circle [radius=2pt]
($(B)!\i/5!(D)$) circle [radius=2pt]
($(D)!\i/5!(C)$) circle [radius=2pt]
($(A)!\i/5!(E)$) circle [radius=2pt]
($(B)!\i/5!(E)$) circle [radius=2pt]
($(C)!\i/5!(F)$) circle [radius=2pt]
($(D)!\i/5!(F)$) circle [radius=2pt]
($(E)!\i/5!(F)$) circle [radius=2pt]};
\end{tikzpicture}
\end{minipage}%
\begin{minipage}{0.3\textwidth}
\centering
\begin{tikzpicture}[scale=0.7]
\coordinate (A) at (-1.5,2);
\coordinate (B) at (1.5,2);
\coordinate (C) at (-1.5,-2);
\coordinate (D) at (1.5,-2);
\coordinate (E) at (0,0.8);
\coordinate (F) at (0,-0.8);
\draw (A)--(B)--(D)--(C)--(A)--(E)--(B)(E)--(F)(C)--(F)--(D);
\draw[color=blue!50!white, thick]
\foreach \i in {1} {
($(A)!\i/5!(B)$)--($(A)!\i/5!(E)$)
}
\foreach \i in {1} {
($(B)!\i/5!(E)$)--($(B)!\i/5!(D)$)
}
\foreach \i in {1} {
($(E)!\i/5!(B)$)--($(E)!\i/5!(F)$)
}
\foreach \i in {1} {
($(F)!\i/5!(C)$)--($(F)!\i/5!(D)$)
}
\foreach \i in {1} {
($(D)!\i/5!(C)$)--($(D)!\i/5!(F)$)
}
\foreach \i in {1} {
($(C)!\i/5!(F)$)--($(C)!\i/5!(A)$)
};
\draw[fill=black] 
(A) circle [radius=3pt]
(B) circle [radius=3pt]
(C) circle [radius=3pt]
(D) circle [radius=3pt]
(E) circle [radius=3pt]
(F) circle [radius=3pt];
\draw[fill=white]
\foreach \i in {1,2,3,4} {
($(A)!\i/5!(B)$) circle [radius=2pt]
($(A)!\i/5!(C)$) circle [radius=2pt]
($(B)!\i/5!(D)$) circle [radius=2pt]
($(D)!\i/5!(C)$) circle [radius=2pt]
($(A)!\i/5!(E)$) circle [radius=2pt]
($(B)!\i/5!(E)$) circle [radius=2pt]
($(C)!\i/5!(F)$) circle [radius=2pt]
($(D)!\i/5!(F)$) circle [radius=2pt]
($(E)!\i/5!(F)$) circle [radius=2pt]};
\end{tikzpicture}
\end{minipage}
\caption{An example of a graph $G$ from the proof of Theorem~\ref{t-claw3}, together with the graphs $G'$ and $G'+F$, from left to right, where the edges of $F$ are coloured blue.}\label{fig-hardness-claw-diamond}
\end{figure}
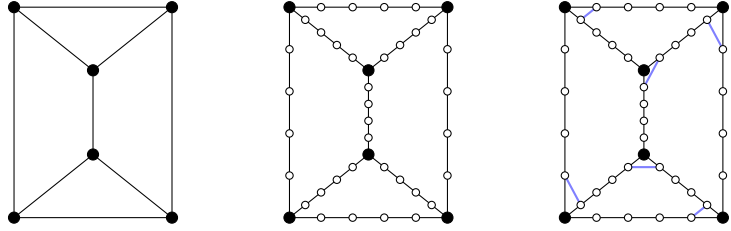

\begin{theorem}\label{t-claw3}
{\sc Perfect Matching Cut} is \NP-complete on the class of probe $(K_{1,3},\overline{2P_1+P_2})$-free subcubic planar graphs (and thus on probe $K_{1,3}$-free graphs).
\end{theorem}

\begin{proof}
We reduce from {\sc Perfect Matching Cut}. Bonnet, Chakraborty and Duron~\cite{BCD23} proved that  {\sc Perfect Matching Cut} is \NP-complete even for cubic bipartite planar graphs.
As such, we may assume that the instance $G$ of {\sc Perfect Matching Cut} is cubic, bipartite and planar.

We subdivide each edge of $G$ four times, that is, we replace each $uv\in E(G)$ with four vertices $y_{uv}^1,\ldots, y_{uv}^4$ and edges $uy_{uv}^1$, $y_{uv}^1y_{uv}^2$, $y_{uv}^2y_{uv}^3$, $y_{uv}^3y_{uv}^4$ and $y_{uv}^4v$. Let $G'$ be the resulting graph. We say that the vertices $y_{uv}^1$ and $y_{uv}^4$ are the {\it intermediate-close} vertices of $G'$. See Figure~\ref{fig-hardness-claw-diamond} for an example.

We claim that $G'$ is a probe $(K_{1,3},\overline{2P_1+P_2})$-free subcubic planar graph. In order to see this, let $N$ be the set of intermediate-close vertices of $G'$. We note that $N$ is an independent set in $G'$. 
As $G$ is cubic, every $u\in V(G)$ has three intermediate-close neighbours in $G'$. We select exactly one edge between two of them to be in $F$, while maintaining planarity; see also Figure~\ref{fig-hardness-claw-diamond}. Hence, $G'+F$ is planar. Moreover, it is readily seen that $G'+F$ is  $(K_{1,3},\overline{2P_1+P_2})$-free. Finally, as intermediate-close vertices in $G'$ have degree~$2$, it follows that $G'+F$ is subcubic.

Le and Telle~\cite{LT22} proved that a graph has a perfect matching cut if and only if the graph obtained from it by subdividing an edge four times has a perfect matching cut.
     Hence, $G$ has a perfect matching cut if and only if $G'$ has a perfect matching cut, which means the theorem is proven. \qed
\end{proof}

\noindent
We recall that a graph is split if and only if it is $(2P_2,C_4,C_5)$-free, and we show:

\begin{theorem}\label{t-split}
For every $d\geq 1$, $d$-{\sc Cut} and {\sc Perfect Matching Cut} are \NP-complete on probe split graphs (and thus on probe $2P_2$-free graphs).
\end{theorem}

\begin{proof}
It is known that for every $d\geq 1$, {\sc $d$-Cut} and {\sc Perfect Matching Cut} are \NP-complete for bipartite graphs. The former statement was shown in~\cite{Mo89} for $d=1$ and in~\cite{FLPR23} for $d\geq 2$. The latter statement was shown in~\cite{LT22}.

Let $G$ be a connected bipartite graph. Let $N$ be one of the two bipartition classes. Note that $N$ is an independent set. Let $F$ consist of all the edges between vertices of $N$. Then $G+F$ is a split graph.
Hence, the theorem follows. \qed
\end{proof}

\noindent
The proof of our final result is similar to the proof of Lucke et al.~\cite{LMPS24} for showing that $d$-{\sc Cut} is \NP-complete for $3P_2$-free graphs for every $d\geq 2$.

\begin{theorem}\label{t-4p1}
For every $d\geq 2$, $d$-{\sc Cut} is \NP-complete on probe $4P_1$-free graphs.
\end{theorem}

\begin{proof}
We first suppose that $d=2$.
We reduce from a restricted variant of the {\sc $3$-Satisfiability} problem. 
Let $X= \{x_1,x_2,\cdots,x_n\}$ be a set of variables, and let  ${\cal C} = \{C_1, C_2, \cdots, C_m\}$ be a set of clauses over $X$. A {\it truth assignment} sets each $x_i$ either true or false. The {\sc $3$-Satisfiability} problem is that of deciding whether $(X,{\cal C})$ has a
{\it satisfying} truth assignment~$\phi$, that is, $\phi$ sets at least one literal true in each $C_i$. Darmann and D\"ocker~\cite{DD21} proved that {\sc $3$-Satisfiability} is \NP-complete even for instances in which:
\begin{enumerate}[(i)]
\item each variable occurs as a positive literal in exactly two clauses and as a negative literal in exactly two other clauses, and
\item each clause consists of three distinct literals that are either all positive or all negative.
\end{enumerate}

\noindent
By (ii), we can write ${\cal C}=\{C_1,\ldots,C_p\}\cup \{D_1,\ldots,D_q\}$ for some $p$ and $q$ with $p+q=n$, where every $C_i$ consists of only positive literals, and every $D_j$ consists of only negative literals.
From $(X,{\cal C})$, we construct a graph $G=(V,E)$. We first introduce a clique $K=\{C_1,\ldots,C_p\}$, a clique $K'=\{D_1,\ldots,D_q\}$ and an independent set $I=\{x_1,\ldots,x_n\}$, such that $K$, $K'$, $I$ are pairwise disjoint and $V=K\cup K'\cup I$.
We add an edge between $x_h$ and $C_i$ if and only if $x_h$ occurs as a literal in $C_i$. Finally, we add an edge between $x_h$ and $D_j$ if and only if $x_h$ occurs as a literal in $D_j$.
See Figure~\ref{f-4p1} for an example. 

\begin{figure}[t]
\centering
\begin{tikzpicture}

\foreach \i  in {1,...,4}{
	\node[rvertex, label = above: $C_\i$](C\i) at (\i, 2.5){};
}

\begin{scope}[shift = {(5,0)}]
\foreach \i  in {1,...,4}{
	\node[bvertex, label = above: $D_\i$](D\i) at (\i, 2.5){};
}\end{scope}

\node[rvertex, label = below: $x_1$](x1) at (2.5, 0){};
\foreach \i  in {2,...,5}{
	\node[bvertex, label = below: $x_\i$](x\i) at (\i+1.5, 0){};
}
\node[rvertex, label = below: $x_6$](x6) at (7.5, 0){};

\draw[edge](C1) -- (x1);
\draw[edge](C1) -- (x2);
\draw[edge](C1) -- (x3);

\draw[edge](C2) -- (x1);
\draw[edge](C2) -- (x3);
\draw[edge](C2) -- (x4);

\draw[edge](C3) -- (x2);
\draw[edge](C3) -- (x5);
\draw[edge](C3) -- (x6);

\draw[edge](C4) -- (x4);
\draw[edge](C4) -- (x5);
\draw[edge](C4) -- (x6);

\draw[edge](D1) -- (x1);
\draw[edge](D1) -- (x2);
\draw[edge](D1) -- (x4);

\draw[edge](D2) -- (x1);
\draw[edge](D2) -- (x3);
\draw[edge](D2) -- (x5);

\draw[edge](D3) -- (x2);
\draw[edge](D3) -- (x4);
\draw[edge](D3) -- (x6);

\draw[edge](D4) -- (x3);
\draw[edge](D4) -- (x5);
\draw[edge](D4) -- (x6);

\draw[dashed] (2,-0.5) rectangle (8,0.3);
\draw[dashed] (0.5,2.2) rectangle (4.5,4.3);
\draw[dashed] (5.5,2.2) rectangle (9.5,4.3);
\node[](k) at ( 0,3.25) {$K$};
\node[](k) at ( 10,3.25) {$K'$};
\node[](k) at ( 1.5,0) {$I$};

\end{tikzpicture}
\caption{\label{f-4p1} An example of a graph $G$ in the proof of Theorem~\ref{t-4p1}, namely when $X = \{x_1, \dots, x_6\}$ and $\mathcal{C} = \{ \{ x_1, x_2, x_3\},\{x_1, x_3, x_4 \},\{x_2, x_5, x_6 \},\{x_4, x_5, x_6 \}\} \cup 
\{\{\overline{x_1}, \overline{x_2},  \overline{x_4}\},\{\overline{x_1}, \overline{x_3},  \overline{x_5}\},  \{\overline{x_2}, \overline{x_4},  \overline{x_6}\},  \{\overline{x_3}, \overline{x_5},  \overline{x_6}\}\}$. For readability the edges inside the cliques $K$ and $K'$ are not shown. The figure is based on a corresponding figure from~\cite{LMPS24}.}
\end{figure}
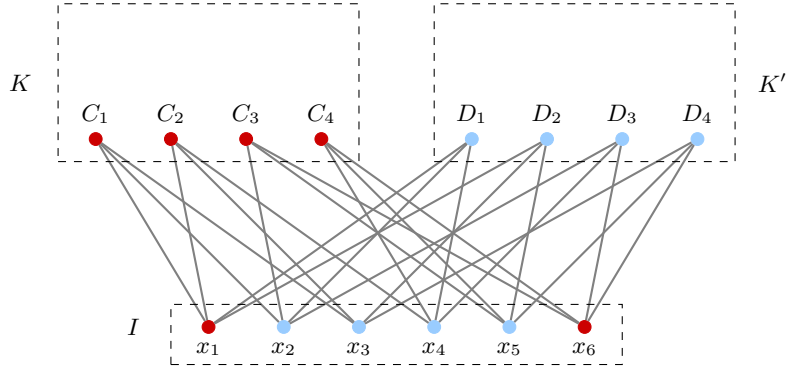

We claim that $G$ is probe $4P_1$-free. That is, we can set $N=I$, and add all edges between the vertices of $I$. This results in a graph~$G'$ that is the union of three cliques $I$, $K$, $K'$, and thus $G'$ is $4P_1$-free. This claim is the only additional claim compared to the proof in~\cite{LMPS24}.

We claim that $(X,{\cal C})$ is a yes-instance of {\sc $3$-Satisfiability} if and only if $G$ has a $2$-cut.

First suppose $(X,{\cal C})$ is a yes-instance of {\sc $3$-Satisfiability}. Then $X$ has a truth assignment $\phi$ that sets at least one literal true in each $C_i$ and in each $D_j$. In $I$, for $h \in \{1,\ldots,n\}$, we colour
the vertex $x_h$ red if $\phi$ sets $x_h$ to be true, and otherwise we colour $x_h$ blue. We colour all the vertices in $K$ red and all the vertices in~$K'$ blue. 

Consider a vertex $x_h$ in $I$. First suppose that $x_h$ is coloured red. 
As each literal appears in exactly two clauses from $\{D_1,\ldots, D_q\}$, we find that $x_h$ has only two blue neighbours (which all belong to $K'$). 
Now suppose that $x_h$ is coloured blue. As each literal appears in exactly two clauses from $\{C_1,\ldots, C_p\}$, we find that $x_h$ has only two red neighbours (which all belong to $K$). 

Now consider a vertex~$C_i$ in $K$, which is coloured red. As $C_i$ consists of three distinct positive literals and $\phi$ sets at least one positive literal of $C_i$ to be true, $C_i$ is adjacent to at most two blue vertices in $I$.
Now consider a vertex~$D_j$ in $K'$, which is coloured blue. As $D_j$ consists of three distinct negative literals and $\phi$ sets at least one negative literal of $D_j$ to be true, we find that $D_j$ is adjacent to at most two red vertices in $I$. 

Now suppose $G$ has a $2$-cut. By Observation~\ref{o-cut-colouring}, this means that $G$ has a red-blue $2$-colouring $f$. We may assume without loss of generality that $|K|=p+1\geq 5$ and $|K'|=q+1\geq 5$ (otherwise we can solve the problem by brute force). Hence, both $K$ and $K'$ are monochromatic. Thus we may assume without loss of generality that $f$ colours every vertex of $K$ red. For a contradiction, suppose $f$ also colours every vertex of $K'$ red. As $f$ must colour at least one vertex of $G$ blue, this means that $I$ contains a blue vertex $x_i$.  
As each variable occurs as a positive literal in exactly two clauses and as a negative literal in exactly two other clauses,
 we now find that a blue vertex, $x_i$, has two red neighbours in $K$ and two red neighbours in $K'$, so four red neighbours in total, a contradiction. We conclude that $f$ must colour every vertex of $K'$ blue.

Recall that every $C_i$ and every $D_j$ consists of three literals. Hence, every vertex in $K\cup K'$ has three neighbours in $I$. As every vertex $C_i$ in $K$ is red, this means that at least one neighbour of $C_i$ in $I$ must be red.
As every vertex $D_j$ in $K'$ is blue, this means that at least one neighbour of $D_j$ in $I$ must be blue.
Hence, setting $x_i$ to true if $x_i$ is red in $G$ and to false if $x_i$ is blue in~$G$ gives us a satisfying truth assignment of $X$. 

\medskip
\noindent
Now, we consider the case when $d\geq 3$. We adjust $G$ as follows. We first modify $K$ into a larger clique by adding, for each $x_h$, a set $L_h$ of $d-3$ new vertices. We also modify $K'$ into a larger clique by adding for each $x_h$, a set $L'_h$ of $d-3$ vertices. For each $h\in \{1,\ldots,n\}$ we make $x_h$ complete to both $L_h$ and to $L_h'$.
Finally, we add additional edges between vertices in $K\cup L_1\cup \ldots \cup L_n$ and vertices in $K'\cup L_1'\cup \ldots \cup L_n'$, in such a way that every vertex in $K$ has $d-2$ neighbours in $K'$, and vice versa. 
The modified graph $G$ is still probe $4P_1$-free, and also still has size polynomial in $m$ and $n$.
The remainder of the proof uses the same arguments as before. \qed
\end{proof}

\section{Conclusions}\label{s-con}

We showed exactly for which graphs~$H$, polynomial results for $d$-{\sc Cut} ($d\geq 1$), {\sc Perfect Matching Cut} and {\sc Maximum Matching Cut} can be extended from $H$-free graphs to probe $H$-free graphs. This yielded complete complexity dichotomies for all three problems on probe $H$-free graphs. 

We note that Theorems~\ref{t-0} and~\ref{t-1} still contain open cases. 
We also propose to study other graph problems on probe $H$-free graphs; so far, systematic studies have only been performed for {\sc Vertex Cover}~\cite{BOPPRL25} and the problems in this paper.

\bibliographystyle{splncs04}

\end{document}